\documentclass[11pt,letterpaper,twocolumn]{article}
\title{Stochastic Kronecker Graph on Vertex-Centric BSP}
\date{}
\usepackage{authblk}
\usepackage[T1]{fontenc}
\usepackage[utf8]{inputenc}
\author[*]{Ernest Ryu}
\author[$\dagger$]{Sean Choi}
\affil[*]{Institute for Computational and Mathematical Engineering, Stanford University}
\affil[$\dagger$]{Department of Computer Science, Stanford University}

\usepackage{qtree,color}
\usepackage[usenames,dvipsnames,svgnames,table]{xcolor}
\usepackage{amsmath, amssymb,amsthm,comment,url}
\usepackage{graphicx,subfig}
\usepackage{algorithmic,algorithm}
\usepackage[dvips,letterpaper,margin=0.7in]{geometry}
\usepackage{biblatex}
\bibliography{bibliography.bib}
\newtheorem{thm}{Theorem}
\newtheorem*{lm}{Lemma}
\newtheorem*{rem}{Remark}
\theoremstyle{remark}

\theoremstyle{remark}

\newcommand\blfootnote[1]{%
  \begingroup
  \renewcommand\thefootnote{}\footnote{#1}%
  \addtocounter{footnote}{-1}%
  \endgroup
}

\begin{document}
\maketitle
\begin{abstract}
Recently Stochastic Kronecker Graph (SKG),
a network generation model,
and vertex-centric BSP,
a graph processing framework like Pregel,
have attracted
much attention in the network analysis community.
Unfortunately the two are not very well-suited for each other
and thus an implementation of SKG 
on vertex-centric BSP
must either be done serially or in an unnatural manner.

In this paper, we present a new network generation model,
which we call Poisson Stochastic Kronecker Graph (PSKG),
that generate edges according to the Poisson distribution.
The advantage of PSKG is that
it is easily parallelizable on vertex-centric BSP,
requires no communication between computational nodes,
and yet retains all the desired properties of SKG.
\blfootnote{Part of this work was done while the second author was visiting LinkedIn.} 
\end{abstract}

\section{Introduction}
With the advent of massive real-world network data and the
computation power to process them,
network analysis is becoming a major topic
of scientific research.
As approaches to model real-world networks,
Stochastic Kronecker Graph (SKG)
\cite{Leskovec:2010:KGA:1756006.1756039}
and its predecessor R\nobreakdash-MAT \cite{Chakrabarti04r-mat:a}
have attracted interest in the
network analysis community
due to its simplicity and its ability to capture
many properties of real-world networks.
As a programming model to process large graphs,
vertex-centric BSP, such as 
Pregel \cite{Malewicz:2010:PSL:1807167.1807184},
Apache Giraph \cite{giraph},
GPS \cite{gps},
and 
Apache Hama \cite{hama},
has become increasingly popular as an alternative
to MapReduce and Hadoop,
which are ill-suited to run massive scale graph algorithms \cite{Malewicz:2010:PSL:1807167.1807184}.

The two, however, are not well-suited for each other. The obvious approach of parallelizing SKG,
which is to generate edges in parallel, is not ``vertex-centric'' in nature and therefore is
unnatural to program and runs inefficiently in vertex-centric BSP.

Therefore we present a new network generation model,
which we call Poisson Stochastic Kronecker Graph (PSKG),
as an alternative.
In this model, the out-degree of each vertex is determined by
independent but non-identical Poisson random variables
and the destination node of the edges are determined in
a recursive manner similar to SKG.

The resulting algorithm, PSKG, is essentially equivalent
to SKG and will therefore retain all the desired properties of it.
Unlike SKG, however, PSKG is embarrassingly parallel in a vertex-centric
manner and therefore is very well-suited for vertex-centric BSP.

\begin{figure}[ht!]
\center
\includegraphics[width=0.5\textwidth]{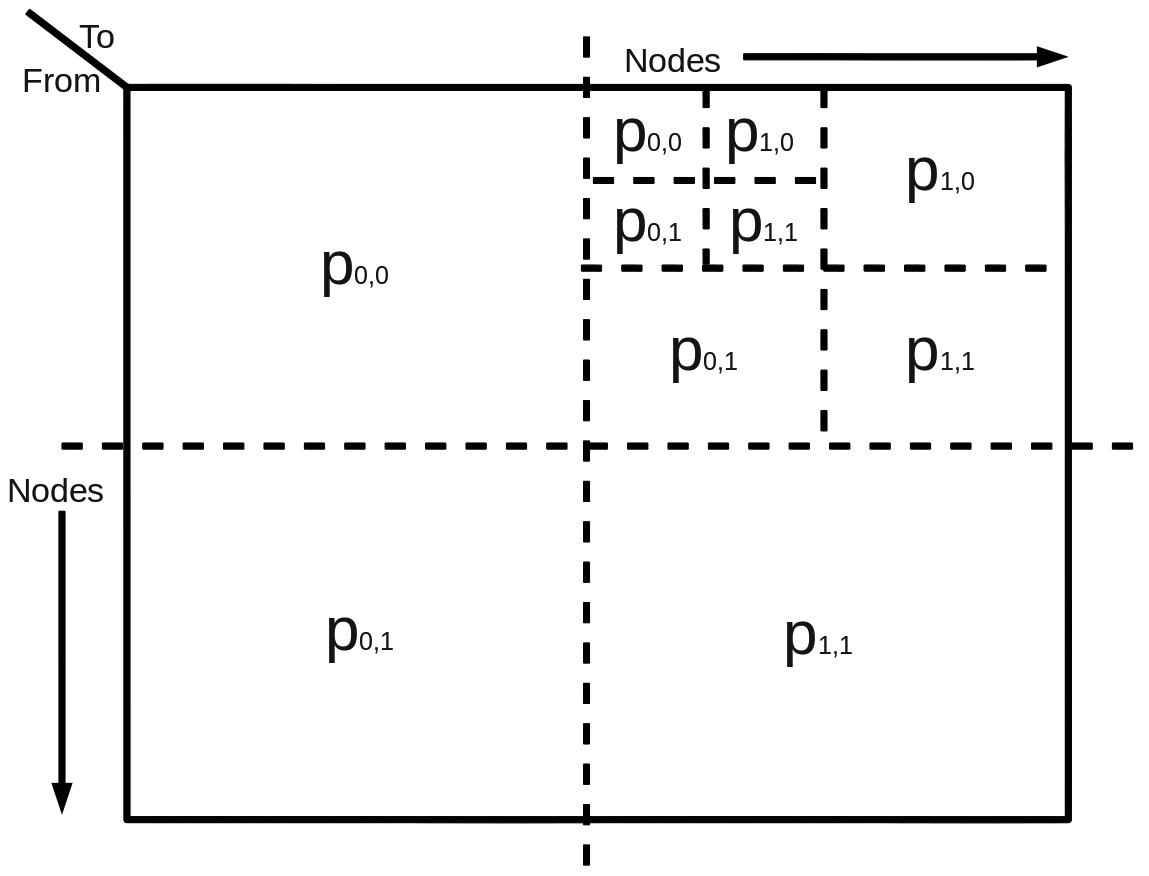}
\caption{Illustration of SKG. At each of the $k$ steps
a sub-region is chosen with probability $p_{i,j}$.}
\label{fig:rmat}
\end{figure}

\section{Theory and Algorithm}
The main result of this paper, PSKG, is presented in Algorithm \ref{alg:pskg}.
We shall, however, start by discussing the original SKG  and an equivalent formulation of it
as this path will motivate PSKG.
We then present the main algorithm.
We conclude this section by discussing issues with load balancing.
\begin{rem}
Throughout this paper we shall use zero-based indexing for vectors and matrices.
\end{rem}

\subsection{Stochastic Kronecker Graph}

Consider the problem of generating a random graph of
$E$ edges and $N=n^{k}$ vertices, where $k\in \mathbb{N}$.
Let 
\[
P=\begin{bmatrix}
p_{0,0}&	\cdots&	p_{0,n-1}\\
\vdots&	\ddots&	\vdots\\
p_{n-1,0}&\cdots & p_{n-1,n-1}
\end{bmatrix}\in \mathbb{R}^{n\times n}
\]
be the ``initiator matrix'' where $p_{i,j}\ge 0$
and $\sum_{i,j}p_{i,j}=1$.

The approach in SKG is to start
off with an empty adjacency matrix and ``drop'' edges into the matrix one at a time.
Each edge chooses one of the $n\times n$ partitions with probability $p_{i,j}$ respectively.
The chosen partition is again subdivided into smaller partitions, and the procedure
is repeated $k$ times until we reach a single cell of the $N\times N$ adjacency matrix and place an edge.
Figure \ref{fig:rmat} illustrates the idea and Algorithm \ref{alg:rmat} makes it concrete.

\begin{algorithm}
\caption{SKG}
\begin{algorithmic} 
\FOR {$i=1,\cdots E$}
\STATE $u=v=0$
\FOR {$j=1,\cdots k$}
\STATE With probability $p_{rs}$ choose subregion $(r,s)$
\STATE $u=nu+r;\,v=nv+s$
\ENDFOR
\STATE Add edge $(u,v)$
\ENDFOR
\end{algorithmic}
\label{alg:rmat}
\end{algorithm}
Let $P_k=P\otimes P\otimes \cdots \otimes P$ be the $k$\nobreakdash-th Kronecker power of $P$.
Then we can interpret SKG as generating $m$ edges independently%
\footnote{The edge generations are not quite independent due to possible ``collisions.''}
where any given edge is $(u,v)$ with probability $(P_k)_{uv}$.
Now we can apply Bayes' rule.
\begin{gather}
\mathbf{P}(\text{edge is }(u,v))
=(P_k\mathbf{1})_u\frac{(P_k)_{uv}}{(P_k\mathbf{1})_u}
\label{eq:bayes}
\\=
\mathbf{P}(\text{destination is }u)
\mathbf{P}(\text{source is }v|\text{destination is }u)
\nonumber
\end{gather}
The decomposition \eqref{eq:bayes} permits us to choose the source node first and then the destination node
rather than simultaneously
and we do this with a recursive algorithm to avoid the explicit construction of $P_k$.
Let
\[
U
=\begin{bmatrix}
\sum_{j=0}^{n-1} p_{0,j}\\
\vdots\\
\sum_{j=0}^{n-1} p_{n-1,j}\\
\end{bmatrix}
\quad
V
=\begin{bmatrix}
\frac{p_{0,0}}{U_0}&	
\cdots&
\frac{p_{0,n-1}}{U_0}\\
\vdots&\ddots&\vdots\\
\frac{p_{n-1,0}}{U_{n-1}}&
\cdots &
\frac{p_{n-1,n-1}}{U_{n-1}}
\end{bmatrix}
\]
and we arrive at Algorithm \ref{alg:rmat2} which is equivalent to 
the original formulation of SKG.
\begin{algorithm}
\caption{Equivalent SKG}
\begin{algorithmic} 
\FOR {$i=1,\cdots E$}
\STATE //Select source node $u$
\STATE $u=0$
\FOR {$j=1,\cdots k$}
\STATE With probability $U_r$ choose subregion $r$
\STATE $u=nu+r$
\ENDFOR
\STATE //Select destination node $v$
\STATE $v=0$; $z=u$
\FOR {$j=1,\cdots k$}
\STATE $l=\mathrm{mod}(z,n)$
\STATE With probability $V_{ls}$ choose subregion $s$
\STATE $v=nv+s;\,z=z/n$ (integer division)
\ENDFOR
\STATE Add edge $(u,v)$
\ENDFOR
\end{algorithmic}
\label{alg:rmat2}
\end{algorithm}

Here we note that the source node selection procedure is (approximately) a multinomial
random variable with parameters $E$ and $U^{[k]}$, where $U^{[k]}$ is the
$k$\nobreakdash-th Kronecker power of $U$.

\subsection{Poisson Stochastic Kronecker Graph}
Due to the following elementary result \cite{Brown1984}
we can replace the source node selection procedure, a multinomial random variable,
with i.i.d. Poisson random variables.

\begin{lm}
Let $X_1,\cdots X_s$ be independent Poisson random variables each with mean $\alpha p_1,\cdots \alpha p_s$,
where $\alpha>0$, $p_1,\cdots p_s\ge 0$, and $\sum ^{s}_{i=1}p_{i}=1$.
Then 
\begin{gather*}
\mathbf{P}\left(X_1=x_1,\cdots X_s=x_s\middle|  \sum^s_{i=1}X_i=m\right)\\
=
\frac{m!}{x_1!\cdots x_s!}p_1^{x_1}\cdots p_k^{x_s}
\end{gather*}
i.e. conditioned on the sum $X_1,\cdots X_s$ is distributed as a multinomial distribution.
\end{lm}
We are finally ready to state the main algorithm of this paper.
Let $E$ be the expected number of total edges
while $P,U,V,k$ are defined the same as before.
\begin{algorithm}
\caption{PSKG}
\begin{algorithmic} 
\STATE Scatter $E,P,U,V,k$
\FOR {Each vertex $u$}
\STATE //Determine out-degree of $u$
\STATE $p=1$; $z=u$
\FOR {$j=1,\cdots k$}
\STATE $l=\mathrm{mod}(z,n)$;  $p=pU_{l}$
\STATE $z=z/n$ (integer division)
\ENDFOR
\STATE Generate $X\sim \mathrm{Poisson}(Ep)$

\STATE //For each edge determine destination vertex
\FOR {$i=1,\cdots X$}
\STATE $v=0$; $z=u$
\FOR {$j=1,\cdots k$}
\STATE $l=\mathrm{mod}(z,n)$
\STATE With probability $V_{ls}$ choose subregion $s$
\STATE $v=nv+s;\, z=z/n$ (integer division)
\ENDFOR

\STATE Add edge $(u,v)$
\ENDFOR

\ENDFOR
\end{algorithmic}
\label{alg:pskg}
\end{algorithm}

PSKG will retain all the desired properties of
SKG graphs. Specifically,
say there is a desired property observed by
SKG graphs of all sizes with probability $1-\varepsilon$.
Then the Poisson SKG graphs will also have the desired property with
probability $1-\varepsilon$ by the following lemma.
\begin{lm}
Let $A$ be an event that occurs with probability $1-\varepsilon$
for SKG graphs of all sizes. Then $A$ will also hold
with probability $1-\varepsilon$ for Poisson SKG
graphs as well.
\end{lm}
\begin{proof}
\begin{gather*}
\mathbf{P}_{Poisson}(A)=
\mathbf{E}\left[
\mathbf{P}\left(A
\middle|\sum^k_{i=1}X_i=m\right)
\right]
\\
=\mathbf{E}\left[
\mathbf{P}_{Multinomial}(A|m)\right]
>
\mathbf{E}[1-\varepsilon]
=1-\varepsilon
\end{gather*}
\end{proof}

\begin{figure*}[ht]
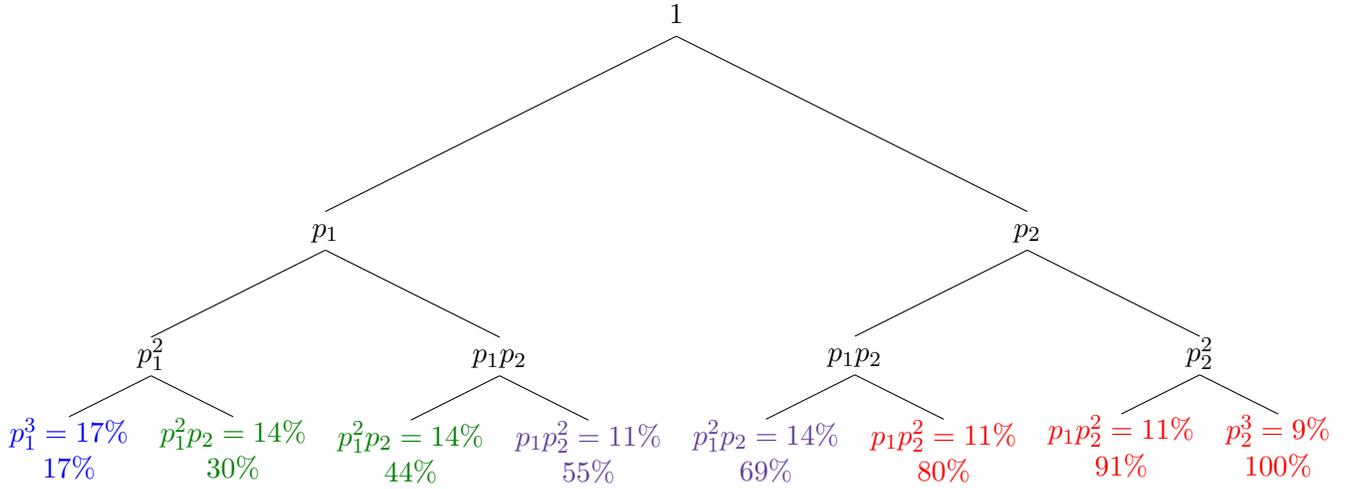

\Tree [.1 [.$p_1$ 
[.$p_1^2$   {\color{Blue}$p_1^3=17\%$}\\{\color{Blue}$17\%$} {\color{Green}$p_1^2p_2=14\%$}\\{\color{Green}$30\%$}  ]
 [.$p_1p_2$  {\color{Green}$p_1^2p_2=14\%$}\\{\color{Green}$44\%$}  {\color{RoyalPurple}$p_1p_2^2=11\%$}\\{\color{RoyalPurple}$55\%$} ]
  ]
  [.$p_2$ 
  [.$p_1p_2$  {\color{RoyalPurple}$p_1^2p_2=14\%$}\\{\color{RoyalPurple}$69\%$} {\color{Red}$p_1p_2^2=11\%$}\\{\color{Red}$80\%$} ]
  [.$p_2^2$  {\color{Red}$p_1p_2^2=11\%$}\\{\color{Red}$91\%$} {\color{Red}$p_2^3=9\%$}\\{\color{Red}$100\%$} ]
   ] ].1
\caption{An example of load balancing where $n=2$, $p_1=0.55$, $p_2=0.45$, $k=3$, and $N_\mathrm{w}=4$.
The second to last line denotes the load of each vertex and the last line denotes the cumulative load.
The 4 processors attempt to take $25\%$ of the total load each.
Consequently processor number {\color{Blue}0}, {\color{Green}1}, {\color{RoyalPurple}2}, and {\color{Red}3}
takes ownership to 
the {\color{Blue}blue}, {\color{Green}green}, {\color{RoyalPurple}purple}, and {\color{Red}red} vertices, respectively.}
\label{fig:tree}
\end{figure*}

\subsection{Probabilistic Load Balancing}
In vertex-centric BSP, where each computational worker
takes ownership to vertices and their outgoing edges, it is not a priori clear how to distribute them;
some vertices have more neighbors than others so assigning an equal number
to each worker will likely result in load imbalance.
As the storage requirement of the graph structure is proportional to the number of neighbors,
 we shall discuss load balancing 
with the goal of distributing the number of edges equally.

Let $N_\mathrm{w}$ denote the total number of workers to balance the load among
and $w_\mathrm{id}=0,1,\cdots (N_\mathrm{w}-1)$ denote the individual processor number.
Let $U^{[k]}$ be the $k$\nobreakdash-th Kronecker power of $U$ and $(U^{[k]})_u$ the
expected load proportion for vertex $u$.
Now we split the set of vertices
into contiguous partitions (contiguous by node numbering) so that
each partition has a total load of about $1/N_\mathrm{w}$ and is owned by one processor.
This procedure is illustrated in Figure \ref{fig:tree}.

To do this partitioning efficiently, however, one must avoid explicitly forming $U^{[k]}$. 
Algorithm \ref{alg:load}%
\footnote{The algorithm is a simplified version specifically for $n=2$.
The generalization to arbitrary $n$ is straightforward.}
 achieves this by traversing the decision tree without explicitly forming it.
 
One legitimate concern of this strategy is that the load is only balanced in expectation
and therefore it is possible that with bad luck
the actual load is highly unbalanced.
However, Theorem \ref{thm:bound} tells us that with high probability
the load imbalance is small.
\begin{thm}
\label{thm:bound}
The $\alpha$\nobreakdash-level confidence interval of the maximum load
over all computational nodes is
$[\frac{E}{N_\mathrm{w}},\frac{E}{N_\mathrm{w}}+\delta]$ where 
\[
\delta=
\sqrt{\frac{2E}{N_\mathrm{w}}}
\sqrt{\log N_\mathrm{w}+\left|\log\left|\log(1-\alpha)\right|\right|}
\]
where $\frac{E}{N_\mathrm{w}}$ is the load under perfect balance.
We can interpret $\delta$ as the degree of load imbalance.
\end{thm}
\begin{proof}
We first make the assumption that the expected load, $E$, is split and distributed perfectly,
i.e. each worker will have a load of $X_i$
where $X_0, \cdots X_{N_\mathrm{w}-1}$ are i.i.d. Poisson random variables with mean $E/N_\mathrm{w}$.
Let $M$ be the upper bound of the confidence interval.
\begin{gather*}
\mathbf{P}\left(\max_i X_i\le M\right)=\mathbf{P}\left(X_1\le M\right)^{N_\mathrm{w}}\\=
F_X^{N_\mathrm{w}}(M)\le
1-\alpha\\
M\le F^{-1}_X\left((1-\alpha)^{1/N_\mathrm{w}}\right)
\approx F^{-1}_X\left(1+\frac{\log(1-\alpha)}{N_\mathrm{w}}\right)\\
\approx
\sqrt{\frac{E}{N_\mathrm{w}}}\Phi^{-1}\left(1+\frac{\log(1-\alpha)}{N_\mathrm{w}}\right)+\frac{E}{N_\mathrm{w}}
\end{gather*}
$(1-\alpha)^{1/N_\mathrm{w}}$ is approximated by its Taylor series
and $X$ is approximated by a normal random variable
$\mathcal{N}(E/N_\mathrm{w},E/N_\mathrm{w})$
given by the central limit theorem.
Abramowitz\cite{AbramowitzStegun64} provides the following bound.
\[
\sqrt{2\pi}
(1-\Phi(x))=
\int^\infty_x e^{-t^2/2}dt
\le 2e^{-x^2/2}
\quad
\text{for }x\ge 0
\]
Using the fact that 
for non-increasing functions $f\le g$ implies $f^{-1}\le g^{-1}$
we arrive at the following.
\[
\Phi^{-1}\left(1-x\right)
\le 
\sqrt{2\log x}
\]
Putting these results together gives the theorem
\end{proof}

\begin{algorithm}
\caption{Load Balancing}
\begin{algorithmic}
\FOR {Each worker}
\STATE $r_\mathrm{low}=
w_\mathrm{id}/N_\mathrm{w}; \,
r_\mathrm{up}=(w_\mathrm{id}+1)/N_\mathrm{w}$
\STATE $b=0;$ //lower bound
\STATE $p_\mathrm{range}=1;$ //probability range
\STATE $u_\mathrm{low}=0$ //lower vertex id
\FOR {$i=1,\cdots k$}
\IF {$r_\mathrm{low}\le b+pp_\mathrm{range}$}
\STATE $p_\mathrm{range}=pp_\mathrm{range};\,u_\mathrm{low}=2u_\mathrm{low}$
\ELSE 
\STATE $b=b+pp_\mathrm{range}$
\STATE $p_\mathrm{range}=(1-p)p_\mathrm{range};\,u_\mathrm{low}=2u_\mathrm{low}+1$
\ENDIF
\ENDFOR
\STATE Repeat above with $u_\mathrm{up}$
\STATE Claim ownership to nodes $u_\mathrm{low}$ to
$u_\mathrm{up}$
\ENDFOR
\end{algorithmic}
\label{alg:load}
\end{algorithm}

\section{Experimental Results}
In this section, we
%first 
demonstrate that PSKG and SKG generate graphs with
essentially the same properties.
As SKG models real world networks well \cite{Leskovec:2010:KGA:1756006.1756039}
the equivalence between PSKG and SKG implies the modeling power of PSKG.
%We then demonstrate the scalability of an Apache Giraph \cite{giraph} implementation of PSKG
%on a distributed cluster.

To generate and analyze the SKG and PSKG graphs
the SNAP library \cite{snap}
and
our own implementation of of PSKG on Apache Giraph \cite{giraphpskg} were
used, respectively.

\subsection{Graph Patterns}
There are several standard graph patterns that are used to
compare the similarity between networks.
In this paper,
we shall use the following patterns: degree distribution,
hop plot, scree plot, and network values.
These choices are motivated by Leskovec's \cite{Leskovec:2010:KGA:1756006.1756039}
work.

\textit{Degree distribution}: The histogram of the nodes' degrees with exponential binning.

\textit{Hop plot}: Number of reachable pairs $r(h)$
within $h$ hops, as a function of the number of hops $h$.

\textit{Scree plot}: Singular values of the graph adjacency matrix versus their rank.

\textit{Network values}: Distribution of the principal eigenvector components versus their rank. 

Figure \ref{fig:graph1} and Figure \ref{fig:graph2} compares the graph patterns of SKG and PSKG.
It is clear that the results are essentially the same.

%\subsection{Scalability}

\begin{figure*}[ht]
\center
\subfloat[Degree distribution]{
\includegraphics[width=0.45\textwidth]{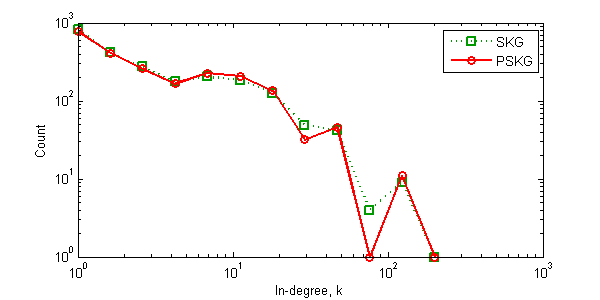}}
\subfloat[Hop plot]{
\includegraphics[width=0.45\textwidth]{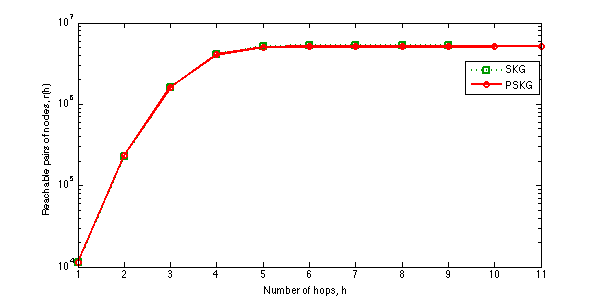}}
\\
\subfloat[Scree plot]{
\includegraphics[width=0.45\textwidth]{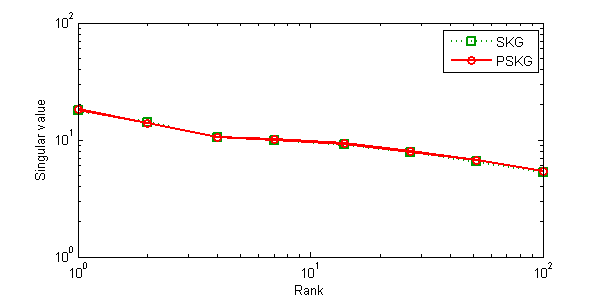}}
\subfloat[``Network value'']{
\includegraphics[width=0.45\textwidth]{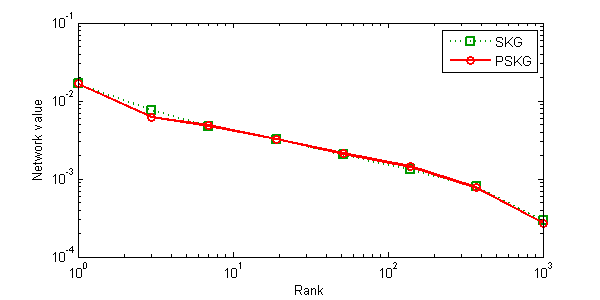}}
\caption{
Graph patterns with parameters 
$n=2$, $k=12$, $E=11400$, and $P=[0.4532,0.2622;0.2622,0.0225]$}
\label{fig:graph1}
\end{figure*}

\begin{figure*}[ht]
\center
\subfloat[Degree distribution]{
\includegraphics[width=0.45\textwidth]{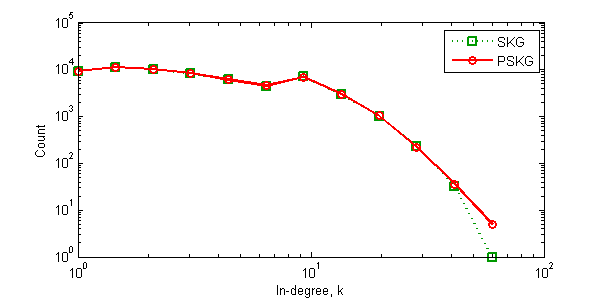}}
\subfloat[Hop plot]{
\includegraphics[width=0.45\textwidth]{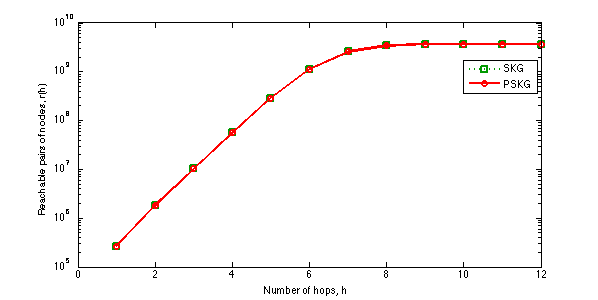}}
\\
\subfloat[Scree plot]{
\includegraphics[width=0.45\textwidth]{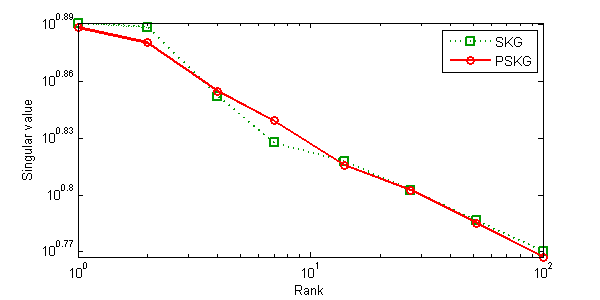}}
\subfloat[``Network value'']{
\includegraphics[width=0.45\textwidth]{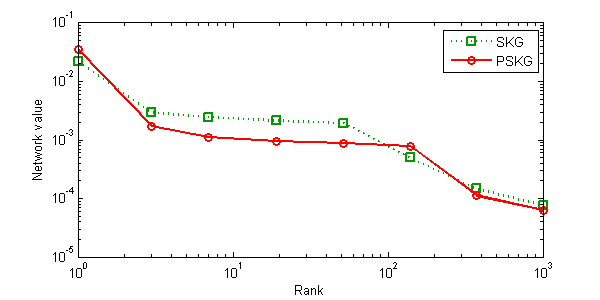}}
\caption{
Graph patterns with parameters
$n=4$, $k=8$, $E=263546$, and $P=[\alpha,\alpha,\alpha,\alpha;\alpha,\alpha,\beta,\beta;\alpha,\beta,\alpha,\beta;\alpha,\beta,\beta,\alpha]$
where $\alpha=0.0861$ and $\beta=0.0231$.
$P$ is the adjacency matrix of a star graph on 4 nodes
(center + 3 satellites)
with the $1$'s are replaced with $\alpha$
and the $0$'s are replaced with $\beta$.}
\label{fig:graph2}
\end{figure*}

\section{Conclusion}
In conclusion, PSKG is a network generation model that is more efficient than and yet
as powerful as SKG. Section 2 and 3 each provide theoretical and empirical evidence
to this statement.

One promising direction of future work is vertex-centric algorithms for model estimation.
There has been much work on SKG model fitting, which should directly apply to PSKG,
but most do not concern vertex-centric parallelism.
It would be interesting to see the efficiency a vertex-centric distributed fitting algorithm
can achieve compared to a serial or MapReduce implementation.

\printbibliography
%\bibliography{bibliography}{}
%\bibliographystyle{plain}
\end{document}